\documentclass[letterpaper]{article}

\usepackage{graphicx}
\usepackage{latexsym}
\usepackage{amsmath}
\usepackage{amsthm}
\usepackage{amssymb}
\usepackage{subfigure}
\usepackage{times}

\newtheorem{theorem}{Theorem}
\newtheorem{proposition}{Proposition}
\newtheorem{lemma}{Lemma}
\newtheorem{definition}{Definition}
\newtheorem{corollary}{Corollary}

\newcommand{\FAT}[1]{\mbox{{$\mathbb{#1}$}}}

\newcommand{\RR}{\FAT{R}}

\renewcommand{\int}[1]{\stackrel{\circ}{#1}}

\def\thm#1{Theorem~\ref{thm:#1}}

\def\ignore#1{}

\begin{document}

\newcommand{\real}{I \hspace{-3pt} R}

\begin{titlepage}

\thispagestyle{empty}

\title{Graph Coalition Structure Generation}
\author{Thomas D. Voice \quad Maria Polukarov \quad Nicholas R. Jennings \\ \small{University of Southampton, UK}\\ \small{\tt\{tdv,mp3,nrj\}@ecs.soton.ac.uk}}
\date{}
\maketitle


\begin{abstract}
\noindent
We give the first analysis of the computational complexity of {\it coalition structure generation over graphs}. Given an undirected graph $G=(N,E)$ and a valuation function $v:2^N\rightarrow\RR$ over the subsets of nodes, the problem is to find a partition of $N$ into connected subsets, that maximises the sum of the components' values. This problem is generally NP--complete; in particular, it is hard for a defined class of valuation functions which are {\it independent of disconnected members}---that is, two nodes have no effect on each other's marginal contribution to their vertex separator. Nonetheless, for all such functions we provide bounds on the complexity of coalition structure generation over general and minor free graphs. Our proof is constructive and yields algorithms for solving corresponding instances of the problem. Furthermore, we derive polynomial time bounds for acyclic, $K_{2,3}$ and $K_4$ minor free graphs. However, as we show, the problem remains NP--complete for planar graphs, and hence, for any $K_k$ minor free graphs where $k\geq 5$. Moreover, our hardness result holds for a particular subclass of valuation functions, termed {\it edge sum}, where the value of each subset of nodes is simply determined by the sum of given weights of the edges in the induced subgraph.
\end{abstract}

\end{titlepage}


\section{Introduction}\label{sec:intro}
\noindent
{\it Coalition structure generation} (CSG) is the equivalent of the complete set partitioning problem~\cite{COMPLETESETPARTITIONING}---one of the fundamental problems in combinatorial optimisation, that has applications in many fields, from political sciences and economics, to operations research and computer science. In particular, it has recently become a major research topic in artificial intelligence and multi-agent systems, as a tool for agents to form effective teams. In a CSG problem, we have a set $N$ of $n$ elements and a valuation function $v:2^N\rightarrow\RR$ over its subsets, and the problem is to divide the given set into disjoint exhaustive subsets (or, {\it coalitions}) $N_1,\ldots,N_m$ so that 
the total sum of values, $\sum_{i=1}^m v(N_i)$, is maximised. Thus, we seek a most valuable partition (or, a {\it coalition structure}) over $N$. 
 

There have been several algorithms developed for CSG. In~\cite{CSGWORSTCASE}, an anytime procedure with worst case guarantees is proposed; however, it reaches an optimal solution after checking all possible coalition structures, and so runs in time $O(n^n)$. On the other hand, algorithms based on dynamic programming (DP)~\cite{COMPLETESETPARTITIONING,Rothkopf1998,Rahwan2008} are guaranteed to obtain an optimal solution in $O(3^n)$. However,   the integer partition (IP) algorithm given in~\cite{Rahwan2007}, although it has the worst case complexity of $O(n^n)$, in practice, is much faster than the DP based algorithms. This algorithm is anytime, and it works by dividing the search space into regions, according to the coalition structure configurations based on the sizes of coalitions they contain, and then performing branch-and-bound search. Furthermore, the improved version of the IP algorithm~\cite{Rahwan2008} uses DP for preprocessing. Alternatively, in~\cite{COMPACTREPRESENTATION}, the authors suggest to utilise compact representation schemes for valuation functions. Indeed, in practice, these functions often display significant structure, and there have been several methods developed to represent them concisely (e.g., by a set of rules)~\cite{Ieong2005,ConSan2004,ConSan2006}. Given this, the problem can be formulated as a mixed integer program (MIP) and solved reasonably well as compared to the IP algorithm that does not make use of compact representations~\cite{COMPACTREPRESENTATION}. Finally, the CSG problem has been also tackled with heuristics methods. In particular,~\cite{Kraus1998} devised a greedy procedure that puts constraints on the possible size of the coalitions formed. This technique, though, does not guarantee that the optimal value will be reached at any point, nor does it give the means of evaluating the quality of the coalition structure selected.

However, all these works assume no structure on the primitive set of elements. This is a considerable shortcoming, as in  various contexts of interest to computer scientists, these elements represent agents (either human or automated) or resources (e.g., machines, computers, service providers or communication lines), which are typically embedded in a social or computer network. Moreover, in many such scenarios those elements which are disconnected, have no effect on each other's performance and potential contribution to a coalition, or may not be able to cooperate at all, due to the lack of communication, coordination or for other reasons. For example, consider a communication network where each edge is a channel, with capacity indicating the amount of information that can be transmitted through it. Any subset of nodes in this network produces a value proportional to the total capacity of the subnetwork induced by these nodes. In such a scenario, any two nodes that are not connected by a direct link in the network, will not affect each other's marginal contribution to any coalition of nodes that separates them. Or, assume that an edge represents a trust link in a reputation system, so that two nodes will only participate in the same coalition if the trust distance given by the length of a path between them, is finite, and suppose that a value of a coalition is given by the number of its mutually trusted members. Then, a contribution of a particular node $i$ will not depend on another node $j$ who trusts some members of the coalition but does not trust $i$ directly. 

Against this background, in this paper we extend the CSG problem to connected sets. More precisely, we consider coalition structures over the node set of a graph, endowed with a valuation function that has the independence of disconnected members. This is formally defined below. 


\subsection{Coalition structure generation over graphs}\label{subsec:GCSG}
\noindent
Given the setting with a finite set of elements $N$ in a connected undirected graph $G=(N,E)$ and a coalition valuation function $v:2^N\rightarrow\RR$ over subsets of $N$, where $v(\emptyset)=0$, we consider a class of coalition structure generation problems over $N$. Accordingly, we make the following definitions.
\begin{definition}\label{def:connectedCS}
For a graph $G = (N, E)$, a coalition structure $\mathcal{C}$ over $N$ is \emph{connected} if the induced subgraph of $G$ over $C$ is connected for all $C \in \mathcal{C}$.
\end{definition}
\begin{definition}\label{def:IDM}
For a graph $G = (N, E)$, a function $v:2^N\rightarrow\RR$ is \emph{independent of disconnected members (IDM)} if for all $i, j \in N$ with vertex separator $C$,
{\small\[
v(C \cup \{i\}) - v(C) = v(C \cup \{i, j\}) - v(C \cup \{j\}).
\]}
\end{definition}
To give an example, suppose that each edge $(i,j)
 \in E$ is associated with a constant weight $v_{i,j}\in\RR$. Then, the coalition valuation function
{\small\[
v(C) = \sum_{(i,j)
 \in E : i,j \in C } v_{i, j}
\]}
has the independence of disconnected members property. We shall term such a function an \emph{edge sum} coalition valuation function. This function is important as it naturally arises in many application scenarios (e.g., communication networks and information systems) and has simple representation.

Note, under Definition~\ref{def:IDM}, if $v(\cdot)$ is IDM and we have two coalitions $B$ and $C$ which are disconnected, then $v(B \cup C) = v(B) + v(C)$. So, for any coalition $C$, its value $v(C)$ is equal to the sum of $v(\cdot)$ over all its connected components. We can deduce that, for any coalition structure $\mathcal{C}$ there exists a coalition structure $\mathcal{D}$ such that $v(\mathcal{C}) = v(\mathcal{D})$ and all coalitions in $\mathcal{D}$ are connected subgraphs. Thus, without loss of generality, we can restrict our attention to coalition structures consisting only of connected subgraphs, which we will call \emph{connected coalition structures}. Moreover, if $G$ is not a connected graph, then we can solve any coalition structure problem over $G$ with an IDM coalition valuation function by finding the optimal coalition structure over each connected component of $G$ and combining the results. The operation of testing connectivity and finding connected components is computationally tractable in polynomial time, and so, without loss of generality, we restrict our attention to connected graphs $G$. Given this, we define a \emph{graph coalition structure generation} (GCSG) problem as follows.
%
%
\begin{definition}\label{def:GCSG}
Given a connected undirected graph $G=(N,E)$ and a coalition valuation function $v:2^N\rightarrow\RR$ which is independent of disconnected members, the \emph{graph coalition structure generation} problem over $G$ is to maximise  $\sum_{C \in \mathcal{C}} v(C)$ for $\mathcal{C}$ a coalition structure over $N$. This problem is equivalent to maximising the same objective function over all \emph{connected} coalition structures.
\end{definition}


\subsection{Our main results}\label{subsec:results}
\noindent
Here, the main results of this paper are summarised. We start by observing that the GCSG problem is NP--complete on general graphs, even for edge sum valuation functions (Section~\ref{sec:general}). Alongside the hardness result, we show that a general instance with $|N|=n$ nodes and $|E|=e$ edges can be solved in time $O\left(n^2 {{e+n}\choose{n}}\right)$ (see \thm{boundGeneral}). For sparse graphs with $e=cn$ edges, where $c$ is a constant, this implies the bound of  $O\left(n^2 y^n\right)$ with a constant $y=\frac{(c+1)^{c+1}}{c^c}$ (see, e.g.,~\cite{CHOOSEBOUND}). 

Given this, we further study special graph classes, namely planar graphs and, more generally, minor free graphs. We give general bounds on the computational complexity of the GCSG problem for these graphs (Section~\ref{sec:minorFree}). Furthermore, we show polynomial time solvability of the GCSG problem for acyclic, $K_{2,3}$ and $K_4$ minor graphs (see~\ref{subsec:smallMinorFree}), and its NP--hardness for planar, and hence, all $K_k$ minor free graphs for $k\geq 5$ (the full proof is given in the appendix). 

To this end, we consider a class of graphs which are guaranteed to contain vertex separators, as defined below.
\begin{definition}\label{def:separator} 
A class of graphs $S$ \emph{satisfies an $f(n)$-separator theorem with constant $\alpha < 1$} if for all $G = (N,E) \in S$ with $|N| = n$ there exist two subgraphs $A, B \subseteq G$ such that $A \cup B = G$, the number of nodes in $A \cap B$ is less than or equal to $f(n)$ and both the number of nodes in $A \setminus B$ and the number of nodes in $B \setminus A$ are less than or equal to $\alpha n$.
\end{definition}
The next theorem is our main technical result.
\begin{theorem}\label{thm:separator}
Suppose a class of graphs $S$ is closed under taking subgraphs and there is an increasing function
$g(n)$ such that for all $G = (N, E) \in S$ with $|N| = n$, graph $G$ has at most $g(n)$ possible connected coalition structures. Suppose further that $S$ satisfies an $f(n)$-separator theorem with constant $\alpha < 1$, and that for any $G \in S$ such a separator can be found in $o(\exp(h(\alpha, n)))$ time, where  $f(n)$ is an increasing $o(n)$ function and
{\small \[
h(\alpha, n) = \sum_{i=0}^{\lfloor \log(n) / |\log(\alpha)| \rfloor} 2 \log(g(f(\alpha^i n)))
+  2 f(\alpha^i n) 
\log\Bigl( \sum_{j=0}^{i-1} f(\alpha^j n) \Bigr).
\]}
Then, for any $1 > \beta > \alpha$, an instance of the graph coalition structure generation problem over a graph from $S$ can be solved in $O(\exp(h(\beta, n)))$ computation steps.
\end{theorem}
The proof and corollaries of Theorem~\ref{thm:separator} are presented in Section~\ref{sec:separator}. Coupled with the separator theorems shown for planar and minor free graphs in~\cite{MINORSEPARATOR,PLANARGRAPHSEPARATOR}, this gives the base to the following results.
\begin{theorem}\label{thm:minorFree}
For any graph $H$ with $k$ vertices and $\gamma > k \sqrt{k} / (1 - \sqrt{2/3})$, an instance of the graph coalition structure generation problem over an $H$ minor free graph $G$ with $n$ nodes requires $O(n^{\gamma \sqrt{n}} )$ computation steps.
\end{theorem}
\begin{theorem}\label{thm:Planar}
For any $\gamma > 2 \sqrt{2} /(1 - \sqrt{2/3})$, a general instance of a graph coalition structure generation problem over a planar graph $G$ with $n$ nodes can be solved in $O(n^{\gamma \sqrt{n}} )$ computation steps.
\end{theorem}
We further obtain polynomial time bounds for small minor free graphs.
\begin{theorem}\label{thm:small}
A GCSG problem can be solved in $O(n^2)$ computational steps for trees with $n$ vertices, and in $O(n^3)$ steps for $K_{2,3}$ or $K_4$ minor free graphs.
\end{theorem}
However, for planar graphs we prove the following hardness result.
\begin{theorem}\label{thm:NPhardPlanar}
The class of edge sum graph coalition structure generation problems over planar graphs is NP--complete. Moreover, a 3-SAT problem with $m$ clauses can be represented by a GCSG problem over a planar graph with $O(m^2)$ nodes. 
\end{theorem}
Note that Theorem~\ref{thm:NPhardPlanar} holds for all $K_k$ minor free graphs where $k \geq 5$, as planar graphs are a special case. This means we should expect it to take time exponential in $\sqrt{n}$ to solve a GCSG problem  over such graphs of size $n$. 
This suggests that the methods given in Theorems~\ref{thm:minorFree} and~\ref{thm:Planar}, which solve these problems in time exponential in $\log (n) \sqrt{n}$, are close to the best possible.

The following sections describe main results and techniques in more detail and contain all auxiliary lemmas, propositions and proofs. 


\section{General Graphs}\label{sec:general}
\noindent
As a first step, we examine the complexity of coalition structure generation over general graphs. For a graph $G=(N,E)$ with a set of nodes $N$ and a set of edges $E$, we denote $|N|=n$ and $|E|=e$. We show the following.
\begin{theorem}\label{thm:boundGeneral}
A general instance of a GCSG problem can be solved in $O\left(n^2 {{e+n}\choose{n}}\right)$ steps, using $O(n^2)$ sized memory.
\end{theorem}
%
\begin{proof}
Every connected coalition structure over $N$ can be expressed as the connected components of some subgraph $G'=(N,E')$ of $G$, where $E'\subseteq E$. Moreover, each connected component has a spanning subtree, so we can restrict our attention to {\it acyclic} subgraphs of $G$. Such a subgraph has at most $n-1$ edges, and so there are at most $\sum_{k=1}^{n-1}{ {e}\choose{k}}$ such subgraphs. Since ${{a}\choose{b}} + {{a}\choose{b-1}} = {{a+1}\choose{b}}$ and ${{a}\choose{b}} \leq {{a+1}\choose{b}}$, this sum is bounded by ${e+n}\choose{n}$. Now, it takes at most $O(n^2)$ steps to determine the connected components of a subgraph, and, thus, there are at most $O\left(n^2 {{e+n}\choose{n}}\right)$ steps needed to check each coalition structure. Finally, it takes at most $O(n^2)$ sized memory to store each coalition as it is checked.
\end{proof}
This is an easy and not particularly promising result, as it may be exponential in $n \log(n)$ and is exponential in $n$ even for sparse graphs. Indeed, the class of graph coalition structure generation problems is NP--hard: it contains the subclass of GCSG problems over complete graphs, which is equivalent to the NP--complete class of standard coalition structure generation problems over node sets. We further show that the problem remains hard even for simple coalition valuation functions, such as edge sum.
\begin{theorem}\label{thm:NPhardGeneral}
The set of GCSGs with edge sum valuation functions is NP--complete.
\end{theorem}
\begin{proof}
We reduce from 3-SAT. 

Suppose we have a 3-SAT problem with variables $x_1, \ldots x_n$, and clauses $C_1 = (l_{1,1}, l_{1,2}, l_{1,3}), \ldots C_m = (l_{m,1}, l_{m,2}, l_{m,3})$. We create graph $G$ as follows. We define a node $n(l_{i,j})$ for each literal $l_{i, j}$, and add to this set of nodes one further node, $s$. We put an edge between $s$ and $n(l_{i,j})$ for every literal $l_{i,j}$, with $v_{s, l_{i,j}} = 1$. For each clause $C_i$ and literals $l_{i,j}$, $l_{i,k}$ we create an edge between $n(l_{i,j})$ and $n(l_{i,k})$ with $v_{n(l_{i,j}), n(l_{i,k})} = -(3m+1)$. Lastly, for each pair of literals $l_{i,j}$ and $l_{k,l}$ such that $l_{i,j}$ represents the negation of $l_{k,l}$ we put an edge between $n(l_{i,j})$ and $n(l_{k,l})$ with $v_{n(l_{i,j}), n(l_{k,l})} = -(3m+1)$.

We claim that the optimal coalition structure for this problem gives value $m$ if and only if the original 3-SAT problem is satisfiable. To see this, note that the total sum of positive edge values is $3m$, and so in any optimal connected coalition structure, there can be no coalition that includes an edge with negative value. So, for all variables $x_i$, there cannot be a literal representing $x_i$ and a literal representing $\overline{x}_i$ in the coalition that contains $s$. Furthermore, the coalition with $s$ in must be connected to at most one of $n(l_{i,1}), n(l_{i,2}), n(l_{i,3})$ for each clause $C_i$. Since all edges with positive values connect to $s$, this means that the optimal coalition structure has value at most $m$.

Suppose there is a coalition structure that gives value $m$. Then, we can satisfy the original 3-SAT problem as follows. For each variable $x_i$, we set $x_i = T$ if and only if there is at least one literal $l_{j,k}$  representing $x_i$, for which $n(l_{j,k})$ is in the same coalition as $s$. Now, for each clause $C_i$, the coalition that contains $s$ must also contain $n(l_{i,j})$ for some $j$. If $l_{i,j}$ represents some variable $x_k$, then $x_k = T$ and so $C_i = T$. Otherwise, $l_{i,j}$ represents $\overline{x}_k$ for some variable $x_k$. This means that no literal representing $x_k$ can be in the same coalition as $s$, and so $x_k = F$ and $C_i = T$ in our allocation. Thus, this allocation satisfies all clauses as claimed.

Conversely, suppose we have some assignment of boolean values to variables which satisfies all clauses in the original 3-SAT problem. Then, we can create a coalition which contains $s$ and exactly one node $n(l_{i,j})$ for each $C_i$, such that the literal $l_{i,j}$ takes the value $T$ under the 3-SAT assignment. If we create singleton coalitions for all other nodes, this results in a coalition structure which has total value equal to the value of the coalition containing $s$. Now, if two nodes $n(l_{i,j})$ and $n(l_{k,l})$ are neighbours, then they cannot both be in the same coalition as $s$. For either $i=k$, in which case only one can be in the same coalition as $s$ by construction, or else $l_{i,j}$ is the negation of $l_{k,l}$ and thus they cannot both be equal to $T$. Hence, the value of this coalition structure is $m$, and the result follows.
\end{proof}
Note that this result can be seen as a corollary of Theorem~\ref{thm:NPhardPlanar} showing the hardness of the edge sum GCSG over planar graphs. However, the proof for planar graphs is much more involved and lengthy, and therefore has been postponed to the end of the paper (see~\ref{subsec:planar}). 


\section{Separator Theorems}\label{sec:separator}
In this section, we prove our main technical result~(\thm{separator}). The proof we give is constructive, and thus yields an algorithm for solving an instance of the problem. We start with an auxiliary lemma below, and then proceed with the proof of the theorem.
\begin{lemma}\label{lem:valuation}
Given a graph $G$ and a coalition valuation function $v(\cdot)$, suppose that $G$ has two edge disjoint subgraphs, $A$ and $B$, which cover $G = A \cup B$ and share some nodes $A \cap B = D$. Suppose further than we have a connected coalition structure $\mathcal{A}$ over $A$, which induces a coalition structure $\mathcal{D}$ over $D$. Then, the coalition valuation function $v^{\mathcal{D}}(\cdot)$ defined over connected coalitions in $B$ as
{\small \[
v^{\mathcal{D}}(F) = v\Bigl(  F \cup \bigcup_{C \in \mathcal{A} : 
C \cap F \neq \emptyset }  C \Bigr)
- \sum_{C \in \mathcal{A} : C \cap F \neq \emptyset } v(C),
\]}
is well defined and independent of disconnected members. Furthermore, if there is a connected coalition structure $\mathcal{B}$ over $B$
and if $\mathcal{C}$ is the coalition structure over $G$ that comes from 
combining $\mathcal{A}$ and $\mathcal{B}$, then
{\small \[
\sum_{C \in \mathcal{C}} v(C) = \sum_{C \in \mathcal{A}} v(C) +
\sum_{C \in \mathcal{B}} v^{\mathcal{D}}(C).
\]}
\end{lemma}
%
\begin{proof}
Since $v(\cdot)$ has the IMD property, if there are two connected coalition structures, $\mathcal{A}$ and $\mathcal{E}$ over $A$ such that both induce $\mathcal{D}$, then for any $F \subseteq B$,
{\small \begin{eqnarray*}
&& v\Bigl(  F \cup \bigcup_{C \in \mathcal{A} : 
C \cap F \neq \emptyset }  C \Bigr)
- \sum_{C \in \mathcal{A} : C \cap F \neq \emptyset } v(C)
\\
&& = v\Bigl(  F \cup \bigcup_{C \in \mathcal{A} : 
C \cap F \neq \emptyset }  C \Bigr)
-
v\Bigl(  \bigcup_{C \in \mathcal{A} : 
C \cap F \neq \emptyset }  C \Bigr),
\\
&& = v\Bigl(  F \cup \bigcup_{C \in \mathcal{E} : 
C \cap F \neq \emptyset }  C \Bigr)
-
v\Bigl(  \bigcup_{C \in \mathcal{E} : 
C \cap F \neq \emptyset }  C \Bigr),
\end{eqnarray*}}
since we have $\{ C \cap F : C \in \mathcal{A} \} = \{ C \cap F : C \in \mathcal{E} \}$ and $A \setminus D$ is disconnected from $F \setminus D$. Hence, $v^{\mathcal{D}}(\cdot)$ is well defined.

Now, suppose there is a connected coalition structure $\mathcal{B}$ over $B$, and $\mathcal{C}$ is the coalition structure over $G$ that comes from combining $\mathcal{A}$ and $\mathcal{B}$. Then, for $C \in \mathcal{C}$, if we let $A_1, \ldots, A_n$ be the connected components of $C \cap A$ and $B_1, \ldots, B_m$ be the connected components of $C \cap B$, we must have that the $A_j$, $j=1,\ldots,n$, are in $\mathcal{A}$ and the $B_i$, $i=1,\ldots,m$, are in $\mathcal{B}$. Furthermore, $v(C \cap A) = \sum_{j=1}^n v(A_j)$, and for each $B_i$,
{\small \begin{eqnarray*}
v^{\mathcal{D}}(B_i) & = & v
\Bigl(B_i \cup \bigcup_{j : A_j \cap B_i \neq \emptyset} A_j  \Bigr)
- \sum_{j : A_j \cap B_i \neq \emptyset} v(A_j),
\\
& = & v\Bigl(B_i \cup \bigcup_{j = 1}^n A_j  \Bigr)
- \sum_{j = 1}^n v(A_j),
\end{eqnarray*}}
since the $A_j$ are disconnected from each other. Thus, for all $i$
{\small \begin{eqnarray*}
v^{\mathcal{D}}(B_i) & = & v(B_i \cup (C \cap A)  )
- v(C \cap A ),
\\
& = & v(B_i \cup C_{i-1}) - v(C_{i-1}),
\end{eqnarray*}}
where $C_0 = (C \cap A)$ and $C_i = (C \cap A) \cup (\cup_{k=1}^{i-1} B_k)$, by the independence of disconnected members. Hence,
{\small \[
\sum_{j = 1}^n v(A_j) +
\sum_{i = 1}^m v^{\mathcal{D}}(B_i) = v(C_m) = v(C).
\]}
Taking the sum over all $C \in \mathcal{C}$ gives the final result.
\end{proof}
%
\begin{proof}[Proof of~\thm{separator}]
Let $\gamma = (\alpha + \beta)/2$. If $g(n)$ is bounded then the problem is trivial and can be solved in $O(1)$ steps. Otherwise, there is some $q$ with $g(n) \geq |\log(4)/ c \log(\gamma)| $ for $n\geq q$, where $c = \frac{1}{2} \Bigl( \frac{1}{|\log(\beta)|} - \frac{1}{|\log(\gamma)|}\Bigr)$. There exists some $k > q$ such that for all $n \geq k$ we have $\alpha n + f(n) < \gamma n$ and 
{\small \[
\Bigl\lfloor \frac{log(n)}{ |\log(\beta)|} \Bigr\rfloor - \Bigl\lfloor  \frac{log(n)}{ |\log(\gamma)|}
\Bigr\rfloor \geq c \log(n)  + \frac{2 log(q)}{|\log(\beta)|} g(q).
\]}
Let us define the function $u(\cdot, \cdot, \cdot)$ for $r \leq n$ as follows:
{\small \[
u(\beta, n,r) = \sum_{i=0}^{\lfloor \log(n) / |\log(\beta)| \rfloor} 2 \log(g(f(\beta^i n)))
+  2 f(\beta^i n) 
\log\Bigl( r + \sum_{j=0}^{i-1} f(\beta^j n) \Bigr).
\]}
Now, for any $n \geq k$ and $d = \frac{\log(4)}{|\log(\gamma)|}$ we have
{\small \[
u(\beta, n, r) - u(\gamma, n, r)  \geq c \log(n) g(q) \geq
d \log(n).
\]}
Furthermore,
{\small 
\begin{eqnarray}
u(\gamma, n, r) - u(\gamma, \alpha n + f(n), r + f(n))
\leq u(\gamma, n, r) - u(\gamma, \gamma n, r + f(n)) \nonumber \\
= 2 \log(g(f(n))) + 2 f(n) \log(r). \label{eqn:kchoice}
\end{eqnarray}}
Suppose there is a constant $K$, and an integer $m > k$ such that for all $G = (N,E) \in S$ with $|N| < m$, for any 
subset $D \subseteq N$, any (not necessarily connected) coalition structure $\mathcal{D}$ on $D$ and any IMD valuation function $v(\cdot)$, we can find the connected coalition structure $\mathcal{C}$ over $G$ which maximises $\sum_{C \in \mathcal{C}} v(C)$ under the constraint that $\mathcal{C}$ induces $\mathcal{D}$ when projected onto $D$, in at most 
{\small \[
K \exp( u(\gamma, n,r)) n^d,
\]}
computation steps, where $n = |N|$ and $r = |D|$. Let us assume that $K$ is large enough so that for all $n$ we have that $s(n) \leq K\exp(h(\gamma, \gamma n))/2$, where $s(n)$ is the maximum number of computation steps required to 
find a separator of size $f(n)$ for a graph with $n$ nodes in $S$.

Now, let us consider a graph $G = (N,E) \in S$ with $|N| = m$, a subset $D \subseteq N$ with a (not necessarily connected) coalition structure $\mathcal{D}$ on $D$ and a IMD valuation function $v(\cdot)$. We will show that we can solve the optimisation problem to maximise $\sum_{C \in \mathcal{C}} v(C)$ over connected coalition structures $\mathcal{C}$ which induce $\mathcal{D}$ over $D$, in at most 
{\small \[
K \exp( u(\gamma, m,r)) m^d,
\]}
computation steps, where $r = |D|$.  Since we can pick an appropriate $K$ for the case where $m=k$, and $u(\gamma, n, 0) n^d \leq h(\beta, n)$, the result will then follow by induction.

First, we find edge disjoint subgraphs $A$ and $B$ such that $A \cup B = G$, $|A \cap B|$ has at most $ f(m)$ nodes and both $A$ and $B$ each have at most $\alpha m +  f(m)$ nodes. Then, for every pair of coalition structures $\mathcal{A}$ over $A \cap (D \cup B)$ and $\mathcal{B}$ over $B \cap (D \cup A)$, we do the following computation. 
We check to see if combining $\mathcal{A}$ and $\mathcal{B}$ together results in a coalition structure that induces $\mathcal{D}$ over $D$. If not, we move on to the next pair of coalition structures. Otherwise, we let $\mathcal{G}$ 
be the coalition structures which $\mathcal{A}$ induces on $D$. Then, we find the connected coalition structure
$\mathcal{E}$ over $A$ which maximises $\sum_{C \in \mathcal{E}} v(C)$ under the constraint that $\mathcal{E}$ induces $\mathcal{A}$ over $A \cap (D \cup B)$, and the connected coalition structure $\mathcal{F}$ over $B$ which maximises $\sum_{C \in \mathcal{F}} v^{\mathcal{G}}(C)$, under the constraint that $\mathcal{F}$ induces $\mathcal{B}$ over $B \cap (D \cup A)$. If either of these problems is infeasible, we move on to the next pair of 
coalition structures. Otherwise, for $\mathcal{E}$ and $\mathcal{F}$, the optimal solutions to these problems, 
we compute
{\small \[
\sum_{C \in \mathcal{E}} v(C) + \sum_{C \in \mathcal{F}} v^{\mathcal{G}}(C).
\]}
This is equal to $\sum_{C \in \mathcal{C}} v(C)$, where $\mathcal{C}$ is the coalition structure over $G$ formed by combining $\mathcal{E}$ and $\mathcal{F}$. We then move on to the next pair of coalition structures, $\mathcal{A}$ and $\mathcal{B}$.

While doing these calculations, we store the current maximum value of $\sum_{C \in \mathcal{C}} v(C)$, along with the relevant $\mathcal{C}$. When the process has completed, this stored maximum $\mathcal{C}$ will be the connected coalition structure which maximises $\sum_{C \in \mathcal{C}} v(C)$ under the constraint that $\mathcal{C}$ induces $\mathcal{D}$ over $D$. Let $\mathcal{C}^*$ be a connected coalition structure which is an optimal solution to this problem. Then, let $E(\mathcal{C}^*)$ be the subgraph of $G$ created from the set of edges which connect two nodes from the same coalition in $\mathcal{C}^*$. We can define a pair of coalition structures $\mathcal{E}^*$ and $\mathcal{F}^*$ by setting coalitions to be the connected components of the induced subgraphs of $E(\mathcal{C}^*)$ over $A$ and $B$ respectively. Since $A \cap B$ is a vertex separator, combining those two induced subgraphs would recover $E(\mathcal{C}^*)$ and hence $\mathcal{C}^*$ can be recovered by combining $\mathcal{E}^*$ and $\mathcal{F}^*$. Now, let $\mathcal{A}^*$ and $\mathcal{B}^*$ be the induced coalition structures of $\mathcal{E}^*$ and $\mathcal{F}^*$ over $A \cap (D \cup B)$ and $B \cap (D \cup A)$, respectively. The result of the above optimisation for the pair of coalition structures $\mathcal{A}^*$ and $\mathcal{B}^*$ must give a value of $\sum_{C \in \mathcal{E}} v(C) + \sum_{C \in \mathcal{F}} v^{\mathcal{G}}(C)$ at least as good as if $\mathcal{E} = \mathcal{E}^*$ and $\mathcal{F} = \mathcal{F}^*$. However, since
{\small \[
\sum_{C \in \mathcal{E}^*} v(C) + \sum_{C \in \mathcal{F}^*} v^{\mathcal{G}}(C)
= \sum_{C \in \mathcal{C}^*} v(C),
\]}
the result must be exactly optimal, and so the resulting output of the above algorithm, $\mathcal{C}$, must be an optimal solution.

Now, to speed up computation, we can choose our pair of coalition structures $\mathcal{A}$ and $\mathcal{B}$ by
first choosing two coalition structures over $(A \cap B) \setminus D$, and then connecting the coalitions from the first to at most one coalition from those induced on $A \cap D$ by $\mathcal{D}$ and connecting the coalitions from the second to at most one coalition from those induced on $B \cap D$ by $\mathcal{D}$. This only excludes possibilities which do not induce $\mathcal{D}$ over $D$ and thus reduces the number of possibilities to check to at most $g(f(m))^2 \max(r,1)^{2  f(m)}$. For each pair $\mathcal{A}$, $\mathcal{B}$, the optimisation takes fewer
computation steps than
{\small \[
2 K  \exp(u(\gamma, \alpha m + f(m), r + f(m)) (\alpha m + f(m))^d,
\]}
which is less than or equal to
{\small \[
2 K  \exp(u(\gamma, \gamma m, r + f(m)) \gamma^c m^d.
\]}
Thereby, the total operation of this algorithm takes at most
{\small \begin{eqnarray*}
&& s(m) +
2 g(f(m))^2 \max(r,1)^{2  f(m)}
K  \exp(u(\gamma, \gamma m, r + f(m)) \gamma^d m^d
\\
&& \leq s(m) + \frac{1}{2} 
g(f(m))^2 \max(r,1)^{2  f(m)}
K  \exp(u(\gamma, \gamma m, r + f(m)) m^d
\\
&& \leq g(f(m))^2 \max(r,1)^{2  f(m)}
K  \exp(u(\gamma, \gamma m, r + f(m)) m^d
\end{eqnarray*}}
computation steps. From (\ref{eqn:kchoice}), this is less than or equal 
to $K  \exp(u(\gamma, m, r ) m^d$,
as required. The main result of~\thm{separator} then follows by induction on $m$.
\end{proof}
The next corollaries follow immediately.
\begin{corollary}\label{cor:sep1}
Suppose a class of graphs $S$ is closed under taking subgraphs and there is an increasing function $g(n)$ such that for all $G = (N, E) \in S$ with $|N| = n$, $G$ has at most $g(n)$ possible connected coalition structures.  Suppose further that $S$ satisfies an $f(n)$-separator theorem with  constant $\alpha < 1$, and that for any $G = (N,E) \in S$ such a separator can be found in at most $o(g(f(n))^2 f(n)^{2f(n)})$ time, where  $f(n)$ is an increasing function that is $o(n)$. Then, for any $1 > \beta > \alpha$, an instance of the GCSG problem over a graph from $S$
can be solved in
{\small\[
O\bigl(n^{(log(2) c f(n)^2 + 2 c f(n) log(c log(n) f(n))} \bigr),
\]}
computation steps, where $c = 1/|\log(\beta)|$.
\end{corollary}
\begin{proof}
The result follows from~\thm{separator} and the observations that
{\small \[
h(\alpha, n) \geq 2 log(g(f(n))) + 2 f(n) log(f(n)),
\]}
and
{\small \[
h(\beta, n) \leq \frac{\log(n)}{|\log(\beta)|} 2 \log(g(f(n))) +
2 f(n) \log\Bigl(  \frac{\log(n)}{|\log(\beta)|} f(n) \Bigr).
\]}
\end{proof}
\begin{corollary}\label{cor:sep2}
Suppose a class of graphs $S$ is closed under taking subgraphs and there is an increasing function
$\mu(n) \leq n^2/2$ such that for all $G = (N, E) \in S$ with $|N| = n$, $|E| \leq \mu(n)$. Suppose further that $S$
satisfies an $f(n)$-separator theorem with constant $\alpha < 1$, and that for any $G = (N,E) \in S$ such a separator can be found in at most $o(2^{2 \mu(f(n))} f(n)^{2f(n)})$ time, where $f(n)$ is an increasing function that is $o(n)$. Then, for any $1 > \beta > \alpha$, an instance of the GCSG problem over a graph from $S$
can be solved in
{\small\[
O\bigl(n^{(2 log(2) c \mu
(f(n)) + 2 c f(n) log(c log(n) f(n))} \bigr),
\]}
computation steps, where $c = 1/|\log(\beta)|$.
\end{corollary}
\begin{proof}
This follows from Corollary~\ref{cor:sep1} and the observation that for all $G = (N,E) \in S$ with $|N| = n$, $|E| \leq \mu(n)$, the number of connected coalition structures over $G$ is less than or equal to the number of subsets of edges, which is bounded by $2^{\mu(n)}$.
\end{proof}


\section{Minor Free and Planar Graphs}\label{sec:minorFree}
\noindent
In this section, we apply~\thm{separator} to obtain computational bounds for minor free and planar graphs. We begin with a technical result.
\begin{proposition}\label{prop:minor}
Suppose a class of graphs $S$ is closed under taking subgraphs and for some constants $K, a, b>0$, for all $G = (N, E) \in S$ with $|N| = n$, $G$ has at most $K \exp(a n^b)$ possible connected coalition structures.  Suppose further that for some constants $L, c > 0$, $S$ satisfies an $L n^c$-separator theorem, with constant $\alpha < 1$, and that for any $G \in S$ such a separator can be found in $o(\exp(2 a L^b n^{b c}) n^{2 c L n^c} )$ time. Then, if $b \leq 1$, for any $\gamma > 2 L c / (1 - \alpha^c)$, the GCSG problem over a graph from $S$ requires $O(n^{\gamma n^c})$ computation steps. If $b > 1$, then for any $\gamma > 2 a L^b/ (1 - \alpha^{b c})$ the problem can be solved in $O(\exp(\gamma n^{b c}))$ computation steps.
\end{proposition}
%
\begin{proof}
Using the terminology of Theorem~\ref{thm:separator}, for all $n$,
{\small \begin{eqnarray*}
h(\alpha, n) & \geq & 2 \log(  g(f(n)) )+ 2 f(n) \log(f(n))\\
& = & 2 \log(K) + 2 a L^b n^{b c} + 2 L n^c ( \log(L) + c \log(n) ).
\end{eqnarray*}}
For all $\beta, n$,
{\small \begin{eqnarray*}
h(\beta, n) & = & \sum_{i=0}^{\lfloor \log(n) / |\log(\beta)| \rfloor} 2 \log(K)
+ 2 a L^b \beta^{b c i} n^{b c}
+  2 L \beta^{c i} n^c 
\log\Bigl( \sum_{j=0}^{i-1} L \beta^{c j} n^c \Bigr)
\\
& \leq & \frac{ \log(n)}{|\log(\beta)|} 2 \log(K)
+ \frac{2 a L^b n^{b c}}{1 - \beta^{b c}}
+ \frac{2 L n^c}{1 - \beta^c}
\log\Bigl( 
\frac{2 L n^c}{1 - \beta^c}
\Bigr)\\
& = & \frac{2 \log(K)}{|\log(\beta)|} \log(n)
+ \frac{2 a L^b }{1 - \beta^{b c}} n^{b c} +
\frac{2 L }{1 - \beta^c} n^c
\Bigl(c \log(n) +  \log \Bigl(   \frac{2 L}{1 - \beta^c} \Bigr)\Bigr).
\end{eqnarray*}}
If we let $\beta' = (\alpha + \beta)/2$, then, if $b \leq 1$, for sufficiently large $n$ we have
{\small \[
h(\beta', n) \leq \frac{2 L c}{1 - \beta^c} n^c \log(n).
\]}
For any  $\gamma > 2 L c / (1 - \alpha^c)$, we can find $\beta > \alpha$ such that $\gamma = 2 L c / (1 - \beta^c)$. Applying Theorem~\ref{thm:separator} with $\beta' = (\alpha + \beta)/2$, and then using the above bound gives us the result for $b \leq 1$. If $b > 1$ then, for sufficiently large $n$ we have
{\small \[
h(\beta', n) \leq \frac{2 a L^b }{1 - \beta^{b c}} n^{b c}.
\]}
For any  $\gamma > 2 a L^b/ (1 - \alpha^{b c})$, we can find $\beta > \alpha$ such that $\gamma = 2 a L^b / (1 - \beta^{b c})$. The result follows by Theorem~\ref{thm:separator} with $\beta' = (\alpha + \beta)/2$ and the bound above.
\end{proof}
We can now prove~\thm{minorFree}.
%
\begin{proof}[Proof of~\thm{minorFree}]
We apply Proposition~\ref{prop:minor} using the main  result in~\cite{MINORSEPARATOR} where it was shown that 
the class of such graphs satisfies a $k \sqrt{k n}$-separator theorem with $\alpha = 2/3$ and the main result in~\cite{COMPLETEMINOREDGES} which showed that any $K_k$ minor free graph of $n$ vertices has at most $q k \sqrt(\log(k)) n$ edges for constant $q < 0.32$. Now, any $H$ minor free graph must be a $K_k$ minor free graph, and hence must have at most $2^{q k \sqrt(\log(k)) n}$ connected coalition structures. So, using the terminology of Proposition~\ref{prop:minor}, we have $K = 1$, $a = \log(2) q k \sqrt(\log(k))$, $b =1$,  $L = k \sqrt{k}$ and $c = 1/2$. Thus, we can solve a general instance of the problem in $O(n^{\gamma n^{1/2}} ))$ for all $\gamma > \frac{2 L c}{1 - \alpha^{c}} = \frac{k \sqrt{k}}{1 - \sqrt{2/3}}$, as required.
\end{proof}
For planar graphs,~\thm{Planar} provides a stronger result. The proof 
follows similar lines as in the proof of~\thm{minorFree} and uses the following additional notation and lemma.
\begin{definition}\label{def:non-crossing}
Let $\mathcal{C}$ be a coalition structure over a set of nodes $C$ which have some ordering $C = \{c_1, c_2, \ldots c_r\}$. We say that $\mathcal{C}$ is \emph{non-crossing} if, for all $1 \leq i < j < k < l \leq r$, if $i, k \in C$ and $j, l \in D$ for $C, D \in \mathcal{C}$, then $C = D$.
\end{definition}
\begin{lemma}\label{lem:inducedcoalitionstructures}
Let $G$ be a planar graph which is embedded in a plane, and let $C$ be a subset of the boundary of the exterior region of $G$. Then, the set of (not necessarily connected) coalition structures over $C$ which can be induced
from connected coalition structures over $G$ minus the edges of $C$ is of size at most $4^{|C|}/2$.
\end{lemma}
%
\begin{proof}
These coalition structures are a subset of the set of non-crossing coalition structures over $C$, using the clockwise ordering of nodes along the boundary path. No connected coalition structure over $G$ can induce a coalition structure over $C$ which isn't non-crossing, or else that would imply the existence of two disjoint paths in $G$ which cross each other in the plane.

Let us define the function $l(\cdot)$ which returns a labelling of the nodes in $C$ for each coalition structure. For $\mathcal{C}$, a coalition structure over $C$, under $l(\mathcal{C})$, we label each node with an $F$ if it is the first node along $C$ in a particular coalition, an $L$ if it is the last node along $C$ in a particular coalition, an $M$ if lies in the middle of a coalition, and an $S$ if it is the sole member of a singleton coalition. The labelling $l(\mathcal{C})$ uniquely defines $\mathcal{C}$ amongst all non-crossing coalition structures. For, given a labelling $l(\mathcal{C})$, we can recover $\mathcal{C}$. We do so by putting each node $u$ which is not labelled $L$ or $S$ in a coalition with the next node $v$ such that $v$ is labelled $L$ and the numbers 
of $L$ and $F$ labelled nodes between $u$ and $v$ are equal. This is similar to the parsing of a string of nested brackets. Since the first node can only be labelled $F$ or $S$ and the last node can only be labelled $L$ or $S$, there are at most $4^{|C|}/2$ such labellings. Hence there are at most $4^{|C|}/2$ such coalition structures over $C$.
\end{proof}
%
\begin{proof}[Proof of~\thm{Planar}]
We apply Proposition~\ref{prop:minor} using the main result in~\cite{PLANARGRAPHSEPARATOR}, which states that planar graphs satisfy a $2 \sqrt{2 n}$-separator theorem with $\alpha = 2/3$, along with Lemma~\ref{lem:inducedcoalitionstructures}. Now, Lemma~\ref{lem:inducedcoalitionstructures} only limits the number of coalition structures that can be induced by a connected coalition structure on the boundary of a graph. However these are precisely the coalitions structures that are considered in the inductive step of Theorem~\ref{thm:separator} when the function $g(n)$ is evaluated. This means that the proof of Theorem~\ref{thm:separator} and Proposition~\ref{prop:minor}, and the corresponding results all hold true for planar graphs, taking $g(n) = 4^n/2$. This corollary follows from Proposition~\ref{prop:minor}, taking $a = \log(4)/2$, $b = 1$, $L = 2 \sqrt{2}$ and $c = 1/2$. A general instance of the problem can then be solved in $O(n^{\gamma n^c})$ computation steps for any $\gamma > \frac{2 L c}{1 - \alpha^c} = \frac{2 \sqrt{2}} {1 - \sqrt{2/3}}$, as required.
\end{proof}
Recall that the class of planar graphs is equivalent to the class of $K_{3,3}$ and $K_5$ minor free graphs. For these graphs, \thm{NPhardPlanar}  shows that the graph coalition structure generation problem is NP--complete, even for simple, edge sum, coalition valuation functions (the proof of the theorem is presented in~\ref{subsec:planar}). However, as we show below, the GCSG over smaller minor free instances can be solved in polynomial time.


\subsection{Small Minor Free Graphs}\label{subsec:smallMinorFree}
\noindent
We now turn to consider $H$ minor free graphs where $H$ is small. The collection of results of this section is summarised in~\thm{small}. 
\begin{lemma}\label{lem:tree}
A graph coalition structure generation problem over a tree $G = (N,E)$
with $|N| = n$ can be solved in $O(n^2)$ computational steps.
\end{lemma}
%
\begin{proof}
The proof proceeds inductively. Suppose that for some $m$, $K$, all graph coalition structure generation problems over trees of $n < m$ nodes can be solved in $K n^2$ computational steps. We assume that $K$ is large enough that for any tree of $n$ nodes, it is possible to find a leaf node, $i$, with a single edge $(i, j)$ and evaluate $v(\{i,j\}) - v(\{i\})$ in $Kn$ steps. Then, given a tree of size $m$, we find a leaf node $i$ with edge $(i, j)$, and evaluate and store $v(\{i,j\}) - v(\{i\})$. Then we complete the graph coalition structure generation problem given by $v(\cdot)$ and $G \setminus i$. We can then extend this to solve the original problem over $G$ by adding $i$ to the coalition which contains $j$ if $v(\{i,j\}) - v(\{i\})$ is positive, or else putting $i$ in coalition $\{i\}$. The total computation time is $Kn^2$.
\end{proof}
The above result is related to results given in~\cite{TREE} regarding coalition structure generation over acyclic graphs. However,~\cite{TREE} does not make the IMD assumption. Their resulting algorithm is more complex than ours and has potentially exponential running time. This is to be expected, as without the independence of disconnected members, the coalition structure generation problem over star networks is necessarily exponential.

The class of trees is equivalent to the class of connected $K_3$ minor free graphs, and so it makes
sense to now consider the classes of $K_4$ and $K_{2,3}$ minor free graphs. We begin with a technical lemma.
\begin{lemma}\label{lem:k4cycle}
Every $2$-connected $K_4$ minor free graph contains a cycle that has at most $2$
nodes with degree greater than $2$, and furthermore, there are no edges between nodes
in the cycle beside those edges which make up the cycle.
\end{lemma}
%
\begin{proof}
Suppose $G$ is a $2$-connected $K_4$ minor free graph. Since it is $2$-connected, it cannot be acyclic. Let $C = \{c_1, c_2, \ldots, c_m\}$ be a cycle in $G$ and let $A_1, \ldots A_r$ be the connected components of $G \setminus C$, and for each $i = 1, \ldots r$ let $B_i$ be equal to all $c_j$ such that $(c_j, k) \in E$ for some $k \in A_i$.
For each $i = 1, \ldots r$ let $D_i$ be the union of $B_i$ and all $A_j$ such that $B_j = B_i$. Let $d(C)$ be the maximum number of nodes in $D_i$ for $i = 1, \ldots r$. We assume without loss of generality that $C$ maximises $d(C)$ over all cycles in $G$. If more than one cycle maximises $d(\cdot)$ then we pick the cycle with fewest nodes.
Now there can be no edges between the nodes in $C$ apart from those that form the cycle itself, otherwise we could
find a cycle whose nodes were a strictly smaller subset of $C$. This cycle would also maximise $d(\cdot)$, which contradicts our choice of $C$.

Now, if $|B_i| = 1$ for any $i$, then letting $B_i = \{c_j\}$, we would have that the removal of $c_j$ from $G$ leaves two disconnected components. This contradicts the $2$-connectedness of $G$. If, for some $i$, $|B_i| \geq 3$, then, letting $B_*$ be three elements of $B_i$, we could create a $K_4$ minor by contracting $A_j$ into the cycle, for $j \neq i$, contracting $A_i$ into one point, and then contracting the nodes in the cycle $c_1, c_2, \ldots c_m$, to the closest nodes in $B_*$. Thus, we must have that $|B_i| = 2$ for all $i$. Let us assume that we have numbered the $D_i$ so that $B_1 = B_2 = \ldots = B_k$ for some $k$ and $B_i \neq B_k$ for $i > k$. Now, if $k = r$ or if $r = 0$ then we are done. Let us suppose otherwise.

For each $i \neq j$, with $B_i \cap B_j = \emptyset$, the two elements of $B_i$ and the two elements of $B_j$ do not occur alternately through the cycle $c_1, c_2, \ldots c_m$. For otherwise, repeatedly contracting along edges of $G$ 
to the nodes of $B_i \cup B_j$ would result in a $K_4$ graph. We define a partial ordering $\geq_C$ over the $B_i$ for $i > k$ where $B_i \geq_C B_j$ if the path along $C$ between the two points of $B_i$ which intersect with fewest points of $B_1$ does intersect with both points of $B_j$. By picking a minimal element according to $\geq_C$, we can find $B_i$ with $i>k$ such that there is a path along $C$ between the two elements of $B_i$ which does not intersect with any other $B_j$ except at $B_j$.

This means we can find a cycle $d_1, \ldots d_s$ such that $\{d_1, \ldots d_s\} \subset \{c_1, \ldots c_m\} \cup B_i \cup A_i$ for some $i > k$, and $\{d_1, \ldots d_s\} \cap (B_j \setminus B_i) = \emptyset$ for all $j \neq i$. This would mean that $d(\{d_1, \ldots d_s\})$ would be strictly bigger than $d(C)$, as all of $D_1$ and at least one point of $B_1$ must lie inside connected components of $G \setminus \{d_1, \ldots d_s\}$ which connect to $\{d_1, \ldots d_s\}$ through $B_i$. Since this is a contradiction, our supposition must be false, and so $k=r$ or $r=0$, and the proof is complete.
\end{proof}
We now prove the complexity of the GCSG over $K_4$ minor free graphs in two steps.
\begin{lemma}\label{lem:k4solve2con}
A GCSG problem over a $2$-connected $K_4$ minor free graph $G = (N,E)$ with $|N| = n$ can be solved in $O(n^3)$ computational steps.
\end{lemma}
%
\begin{proof}
From Lemma~\ref{lem:k4cycle} we know that $G$ contains a cycle with at most $2$ nodes with degree greater than $2$. We can find such a cycle by finding all nodes with degree equal to $2$, deriving the paths these nodes and their neighbours induce, and searching to find a pair of these paths which have the same endpoints. This will take $O(n^2)$ steps at most.

Now, suppose we have found such a cycle $C = \{c_1, \ldots c_r\}$. For each edge in $C$, $(c_j, c_{j+1})$, we compute a value $v_j = v(\{c_j, c_{j+1}\})$, letting $v_r = v(\{c_r, c_1\})$. By independence of disconnected members, the value of any coalition that is a subset of $C$ is equal to the sum of $v_j$ for the all $j$ such that  the corresponding edge is contained within the coalition.

If all nodes have degree $2$ then $G = C$ and we can solve it easily. The optimal coalition structure is either $C$ or the structure that arises when each pair of consecutive nodes in $C$ lie in the same coalition if and only if the corresponding $v_j$ is positive. These two possibilities are easily checked in $O(n^2)$ time. Otherwise, there must be two nodes with degree greater than $2$ (by the $2$-connectedness of $G$). Let $c_1$ and $c_i$ be those nodes. Without loss of generality we can assume that $i > 2$.

If $i = r$ then we can reduce the problem to an induced connected coalition formation problem over $G \setminus \{c_2, \ldots c_{i-1}\}$ in the following manner. If $v_j$ are positive for all $j=1, \ldots i-1$, then we add $\sum_{j=1}^{i-1} v_j$ to the value of any coalition which contains both $c_1$ and $c_j$, and $\sum_{j=1}^{k-1} v_j$ to any coalition which contains $c_1$ but not $c_i$ and $\sum_{j=k+1}^{i-1}$ to any coalition which contains $c_i$ but not $c_1$, where $k$ is such that $v_k = \min_{j=1}^{i-1} v_j$. The optimal coalition structure over $G \setminus \{c_2, \ldots c_{i-1}\}$ can then be extended to an optimal coalition structure over $G$ by adding $P$ to any coalition that contains both $c_1$ and $c_i$, $\{c_2, \ldots c_k\}$ to any coalition that contains $c_1$ but not $c_i$ and $\{c_{k+1}, \ldots c_i\}$ to any coalition that contains $c_i$ but not $c_1$. If $v_j$ is negative for some $1 \leq j < i$ then we just add the sum of all positive $v_j$ for $j = 1, \ldots i-1$ to any coalition that contains $c_1$. The optimal coalition structure over $G \setminus \{c_2, \ldots c_{i-1}\}$ can then be extended to an optimal coalition structure over $G$ by ensuring that for all $j$, $c_j$ and $c_{j+1}$ are in the same coalition if and only if $v_j$ is positive. Note, $G \setminus \{c_2, \ldots c_{i-1}\}$ is $K_4$ minor free, since it is a minor of $G$. Furthermore, it is $2$-connected since, any cycle not equal to $C$ which passes through $C$ must pass through $c_1$ and $c_i$, and so there must exist a cycle which passes through the same points in $G \setminus C$ and uses edge $(c_1, c_i)$.

If $i \neq r$ then we can reduce the problem to an induced graph coalition structure generation problem over the graph $G'$ that is formed by adding the edge $(c_1, c_i)$ to $G \setminus (C \setminus \{c_1, c_i\})$. Let $P_1$ be the path $\{c_1, c_2, \ldots c_i\}$ and let $P_2$ be the path $C \setminus (P_1 \setminus \{c_1, c_i\})$. There is only one coalition structure over $P_1$ in which $c_1$ and $c_i$ lie in the same coalition. If $v_j$ is positive for $j=1, \ldots i-1$ then the optimal coalition structure over $P_1$ in which $c_1$ and $c_i$ do not lie in the same coalition is $\{c_1, \ldots c_k\}, \{c_{k+1}, \ldots c_i\}$ where $k$ is such that $v_k = \min_{j=1}^{i-1} v_j$. If $v_j$ is negative for some $1 \leq j < i$, then the optimal coalition structure over $P_1$ in which $c_1$ and $c_i$ do not lie in the same coalition is the coalition structure where for all $j = 1, \ldots i-1$ $c_j$ and $c_{j+1}$ lie in the same coalition if and only if $v_j$ is positive. Similarly we can find the optimal coalition structures over $P_2$ which have $c_1$ and $c_i$ in the same coalition and in different coalitions. By testing out the four combinations of these coalition structures over $P_1$ and $P_2$ we can find the optimal coalition structures over $C$ which have $c_1$ and $c_i$ in the same coalition and in different coalitions. Let $v_+$ be the value of the optimal coalition structure over $C$ in which $c_1$ and $c_i$ lie in the same coalition, and let $v_-$ be the value of the optimal coalition structure over $C$ in which $c_1$ and $c_i$ lie in different coalitions. Then,  we create a value function for coalitions over $G'$ by modifying $v(\cdot)$ so that for any coalition which contains both $c_1$ and $c_i$, we add $v_+$, and for any coalition which contains $c_1$ but not $c_i$, we add $v_-$. We can then extend an optimal coalition structure over $G'$ to an optimal coalition structure over $G$ by combining it with the appropriate optimal coalition structure over $C$. Note, $G'$ is a minor of $G$ and so must be $K_4$ minor free. It is also $2$-connected, as any cycle in $G$ not equal to $C$ which passes through $C$ must pass through $c_1$ and $c_i$, and so there must exist a cycle in $G'$ which passes through the same points in $G \setminus C$ and uses the edge $(c_1, c_i)$.

Thus, in $O(n^2)$ time we can reduce our problem over $G$ to a problem over a graph with strictly fewer nodes. Hence, by induction, the total time to solve the problem over $G$ is $O(n^3)$.
\end{proof}
The next proposition then follows from Lemmas~\ref{lem:tree} and~\ref{lem:k4solve2con}.
\begin{proposition}\label{prop:k4solve}
A graph coalition structure generation problem over a $K_4$ minor free graph $G = (N,E)$
with $|N| = n$ can be solved in $O(n^3)$ computational steps.
\end{proposition}
%
\begin{proof}
Lemma~\ref{lem:tree} allows us to solve for any acyclic graph in $O(n^2)$ steps. Since it takes $O(n^2)$ steps to test for acyclicity, we can begin our solving algorithm by checking whether $G$ is acyclic and if it is, solving it using the process described in the proof for Lemma~\ref{lem:tree}.

If the graph is acyclic, then the next step will be to split the problem into smaller, independent problems. For each $i \in N$ we test whether or not $G \setminus \{i\}$ is disconnected. If $G \setminus \{i\}$ is disconnected then we split $G$ into the induced subgraphs formed by combining each connected component of $G \setminus \{i\}$ with $\{i\}$. We then repeat the process over each of these subgraphs. However, once a node has been tested as to whether or not its removal disconnects $G$, then it is no longer necessary to test it again. For a node disconnects a graph if and only if it is not on any cycles in that graph. This property is conserved by taking subgraphs as described since, by definition, this process preserves all cycles. Thus, in $O(n^3)$ steps we can find subgraphs $A_1, A_2, \ldots A_r$ such that, $\cup_{i=1}^l A_i = G$ and each $A_i$ is $2$-connected.

We can solve the graph coalition structure generation problem by solving the induced problem over each $A_i$ for $i=1, \ldots r$ and then combining the resulting optimal coalition structures together. This will be optimal for the problem over $G$ by independence of disconnected members.

Now, if were were to create a graph over $l$ nodes where we include the edge $(i, j)$ if and only if $A_i \cap A_j \neq \emptyset$, then, by the definition of the $A_i$, this graph would be acyclic, and would thus have $r-1$ edges. Moreover, if $|A_i \cap A_j| \leq 1$ for all $i$ and $j$, as otherwise we would have that $A_i \cup A_j$ were $2$-connected. So $\sum_{i=1}^r |A_i| \leq n + r -1 \leq 2 n$. Since the cubic function is sub-additive, by Lemma~\ref{lem:k4solve2con}, the solution process as described takes at most $O(n^3)$ time.
\end{proof}
Finally, we show the result regarding $K_{2,3}$ minor free graphs.
\begin{proposition} \label{prop:k23solve}
A graph coalition structure generation problem over a $K_{2,3}$ minor free graph $G = (N,E)$
with $|N| = n$ can be solved in $O(n^3)$ computational steps.
\end{proposition}
%
\begin{proof}
Let $H$ be the graph that is formed by taking a $K_4$, and removing an edge $(n_1, n_2)$, adding a node $n_3$ and adding edges $(n_1, n_3)$ and $(n_2, n_3)$. So $H$ is essentially a $K_4$ where one edge has been replaced by a path of length $3$. Since $K_{2,3}$ is a minor of $H$, any $K_{2,3}$ minor free graph is an $H$ minor free graph. We will show that we can solve any $H$ minor free graph in $O(n^3)$ time.

We first claim that any $H$ minor free graph can only have a $K_4$ minor if that minor is a subgraph. For, suppose $G$ has a $K_4$ minor but no $K_4$ subgraph. This minor arises from at least one edge contraction on a subgraph of $G$. Let us assume that we pick the minimal subgraph of $G$ such that a sequence of edge contractions will yield a $K_4$, and then consider the minor $G'$ which occurs before the last edge contraction. The graph $G'$ has five nodes and has two nodes $n_1, n_2$ such that the edge $(n_1, n_2) \in G'$ and contracting along $(n_1, n_2)$ yeilds a $K_4$. This means that $\{n_3, n_4, n_5\} = G' \setminus \{n_1, n_2\}$ is a $K_3$, and the union of the neighbours of $n_1$ and the neighbours of $n_2$ must be all of $\{n_3, n_4, n_5\}$. If either $n_1$ or $n_2$ have no neighbours in $\{n_3, n_4, n_5\}$ then the edge contraction is not necessary, and instead the extra node could be removed by starting with a smaller subgraph before the edge contractions. This contradicts our choice of subgraph. So both $n_1$ and $n_2$ has at least one neighbour in $\{n_3, n_4, n_5\}$. Suppose $n_1$ and $n_2$ do not share a neighbour. Without loss of generality, say $(n_1, n_3), (n_2, n_4) \in G'$. Then $n_5$ must have a neighbour in $\{n_1, n_2\}$,
and again we say without loss of generality, $(n_1, n_5) \in G'$. Then $G'$ is isomorphic to $H$. Now suppose $n_1$ and $n_2$ do share a neighbour. Without loss of generality, say $(n_1, n_3), (n_2, n_3) \in G'$. Suppose $n_4$ and $n_5$ do not have the same neighbour in $\{n_1, n_2\}$, say without loss of generality $(n_1, n_4), (n_2, n_5) \in G'$. Then $G'$ is isomorphic to $H$. Now suppose $n_4$ and $n_5$ share a neighbour in $\{n_1, n_2\}$, say without loss of generality $(n_1, n_4), (n_1, n_5) \in G'$. Then, removing edge $(n_1, n_3)$, and edges $(n_2, n_4)$ and $(n_2, n_5)$ if either of them are in $G'$, leaves a graph isomorphic to $H$.

Thus, if a graph $G$ is $H$ minor free, it can only have a $K_4$ minor if that minor occurs as a subgraph of $G$. Moreover, if such $G$ has a subgraph $G'$ which is isomorphic to $K_4$, then the removal of the edges of $G'$ from $G$ should disconnect the nodes of $G'$ from each other. For otherwise, there would exist two nodes of $G'$ and a path between them that was disjoint from $G'$. The union of that path and $G'$ would have a $H$ minor.

So, if we follow the first step of the procedure described in the proof of Proposition~\ref{prop:k4solve}, and split $G$ up into $2$-connected subgraphs, then every resulting subgraph will either be isomorphic to $K_4$ or $2$-connected and $K_4$ minor free. The subgraphs isomorphic to $K_4$ can be solved in $O(1)$ steps. The other subgraphs can be solved using the technique described in Lemma~\ref{lem:k4solve2con}. As argued in the proof of Proposition~\ref{prop:k4solve}, the total number of nodes across all subgraphs is bounded by $2n$, and so the entire process takes $O(n^3)$ time. 
\end{proof}


\subsection{Planar Graphs}\label{subsec:planar}
\noindent
Here we prove NP-hardness result for larger minor-free graphs. 
\begin{proof}[Proof of~\thm{NPhardPlanar}]
Suppose we have a 3-SAT problem with clauses $C_1, \ldots C_m$. We will construct an edge sum graph coalition structure generation problem over a planar graph of $O(m^2)$ nodes which, when solved, reveals a solution to the 3-SAT problem if one exists. We first define some components.

The first component is given in Figure~\ref{fig:symbol2}. We will use the symbol in Subfigure~\ref{sfig:symbol2s} to represent three nodes that surround a subgraph with edge values given in Subfigure~\ref{sfig:symbol2ev}. If this is a subgraph of an edge sum problem graph, then the contribution these edge values make to the valuation of a coalition structure is at most $3$, with equality only if the induced structure over the three outer nodes is that given in Subfigure~\ref{sfig:symbol2state1} or that in Subfigure~\ref{sfig:symbol2state2}. If the induced coalition structure over these three nodes is not one of these two structures, then the contribution will be less than $3$. We similarly describe two more triangular components in Figures~\ref{fig:symbol3} and~\ref{fig:symbol4}, and a double line component in Figure~\ref{fig:symbol5}.
\begin{figure}[ht]
\subfigure[Component edge values]{
\label{sfig:symbol2ev}
\includegraphics{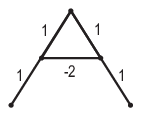}}
\subfigure[Symbol]{
\label{sfig:symbol2s}
\includegraphics{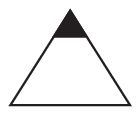}}
\subfigure[Optimal induced structure 1]{
\label{sfig:symbol2state1}
\includegraphics{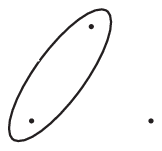}}
\subfigure[Optimal induced structure 2]{
\label{sfig:symbol2state2}
\includegraphics{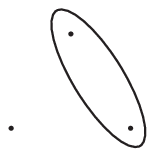}}
\caption{Component of edge sum problem.}\label{fig:symbol2}
\end{figure}
\begin{figure}[ht]
\subfigure[Component edge values]{
\label{sfig:symbol3ev}
\includegraphics{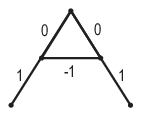}}
\subfigure[Symbol]{
\label{sfig:symbol3s}
\includegraphics{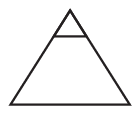}}
\subfigure[Optimal induced structure 1]{
\label{sfig:symbol3state1}
\includegraphics{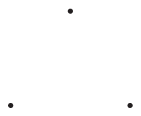}}
\subfigure[Optimal induced structure 2]{
\label{sfig:symbol3state2}
\includegraphics{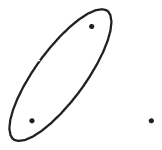}}
\subfigure[Optimal induced structure 2]{
\label{sfig:symbol3state3}
\includegraphics{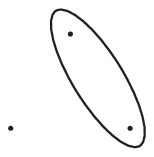}}
\caption{Component of edge sum problem.}\label{fig:symbol3}
\end{figure}
\begin{figure}[ht]
\subfigure[Component edge values]{
\label{sfig:symbol4ev}
\includegraphics{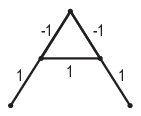}}
\subfigure[Symbol]{
\label{sfig:symbol4s}
\includegraphics{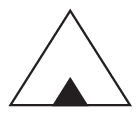}}
\subfigure[Optimal induced structure]{
\label{sfig:symbol4state1}
\includegraphics{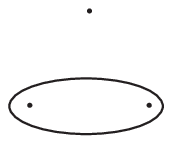}}
\caption{Component of edge sum problem.}\label{fig:symbol4}
\end{figure}
\begin{figure}[ht]
\subfigure[Component edge values]{
\label{sfig:symbol5ev}
\includegraphics{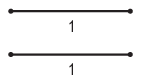}}
\subfigure[Symbol]{
\label{sfig:symbol5s}
\includegraphics{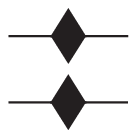}}
\subfigure[Optimal induced structure]{
\label{sfig:symbol5state1}
\includegraphics{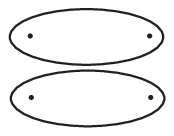}}
\caption{Component of edge sum problem.}\label{fig:symbol5}
\end{figure}
We also describe a last component in Figure~\ref{fig:symbol1}, which we construct out of six copies of the component described in Figure~\ref{fig:symbol2}. For the three points labelled $A, B, C$, there are two induced coalition structures given in Subfigures~\ref{sfig:symbol1state1} and~\ref{sfig:symbol1state2}, for which the contribution of the edge values in the component is maximal.
\begin{figure}[ht]
\subfigure[Component construction]{
\label{sfig:symbol1cons}
\includegraphics{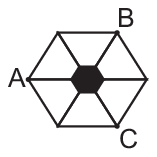}}
\subfigure[Symbol]{
\label{sfig:symbol1s}
\includegraphics{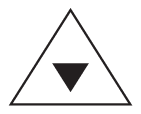}}
\subfigure[Optimal induced structure]{
\label{sfig:symbol1state1}
\includegraphics{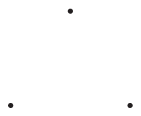}}
\subfigure[Optimal induced structure]{
\label{sfig:symbol1state2}
\includegraphics{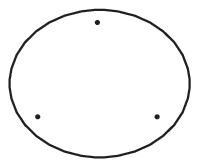}}
\caption{Component of edge sum problem.}\label{fig:symbol1}
\end{figure}
The planar graph edge sum problem we construct will be created from these components. For a graph consisting of these components, we will say that a coalition structure is locally optimal if the induced structure over every component is optimal for that component. Every locally optimal coalition structure is optimal, however it is not
guaranteed that such a structure exists.

We will now describe some constructs which are made from the above described components. The first is given in Figure~\ref{fig:construct1}. It is such that in any locally optimal coalition structure, nodes $X$ and $Y$ are always in the same coalition and the pair of nodes labelled $A$ lie in the same coalition if and only if the pair of nodes labelled $B$ lie in the same coalition.
\begin{figure}[ht]
\subfigure[Construct structure]{
\label{sfig:construct1cons}
\includegraphics[scale = 0.8]{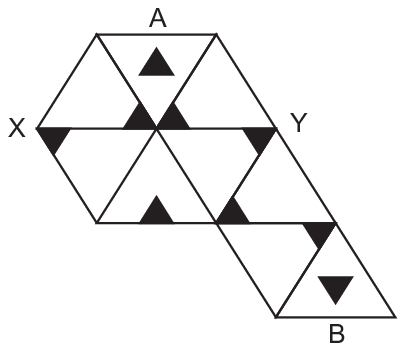}}
\subfigure[Locally optimal structure]{
\label{sfig:construct1state1}
\includegraphics[scale = 0.8]{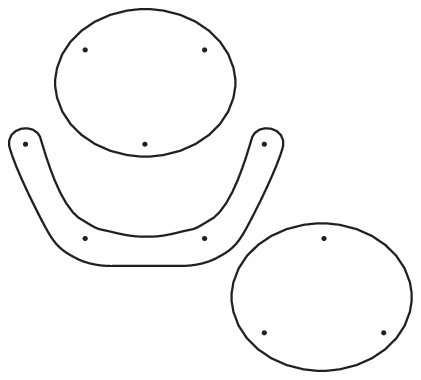}}
\subfigure[Locally optimal structure]{
\label{sfig:construct1state2}
\includegraphics[scale = 0.8]{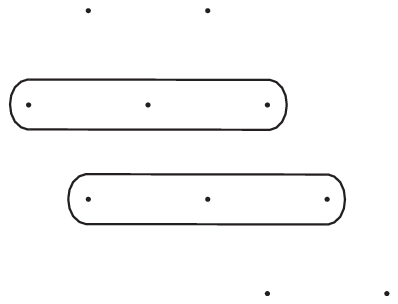}}
\caption{Edge sum construct.}\label{fig:construct1}
\end{figure}
The second and third constructs are given in Figures~\ref{fig:construct3} and~\ref{fig:construct4}. In the second construct, under a locally  optimal coalition structure, if the pair of nodes labelled $A$ are together in the same coalition, then the pair of nodes labelled $B$ are in the same coalition, and similarly for the pair of nodes labelled $C$. If the pair of nodes labelled $A$ are not in the same coalition, then the pair of nodes labelled $B$ are not in the same coalition, and similarly for the pair of nodes labelled $C$. The third construct is similar, except that under a locally optimal coalition structure, the state of whether or not the pair of nodes labelled $C$ are in the same coalition as each other is the opposite to the state of whether or not the pair of nodes labelled $A$ are in the same coalition as each other.
\begin{figure}[ht]
\subfigure[Construct structure]{
\label{sfig:construct3cons}
\includegraphics[scale = 0.8]{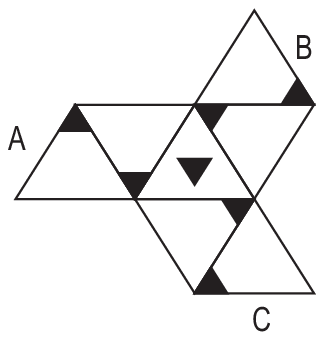}}
\subfigure[Locally optimal structure]{
\label{sfig:construct3state1}
\includegraphics[scale = 0.8]{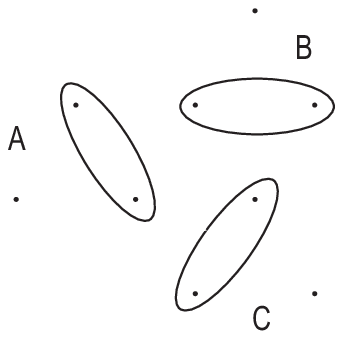}}
\subfigure[Locally optimal structure]{
\label{sfig:construct3state2}
\includegraphics[scale = 0.8]{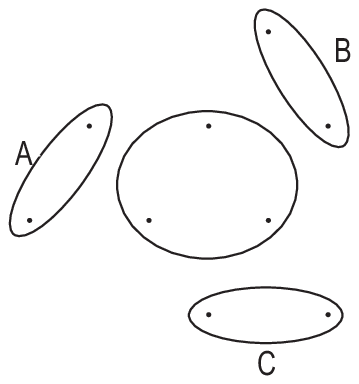}}
\caption{Edge sum construct.}\label{fig:construct3}
\end{figure}
\begin{figure}[ht]
\subfigure[Construct structure]{
\label{sfig:construct4cons}
\includegraphics[scale = 0.8]{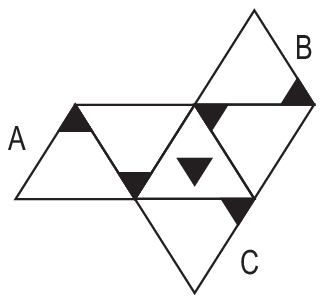}}
\subfigure[Locally optimal structure]{
\label{sfig:construct4state1}
\includegraphics[scale = 0.8]{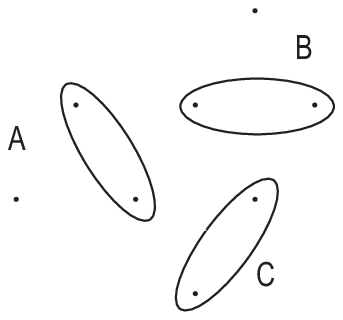}}
\subfigure[Locally optimal structure]{
\label{sfig:construct4state2}
\includegraphics[scale = 0.8]{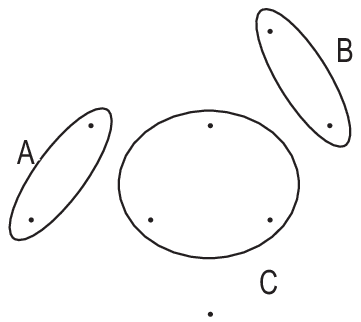}}
\caption{Edge sum construct.}\label{fig:construct4}
\end{figure}
The last construct is given in Figure~\ref{fig:construct2}. Subfigures~\ref{sfig:construct2Astate1}--\ref{sfig:construct2Cstate2} show locally optimal coalition structures over three different parts of the construct. Any coalition structure which induces any combination of these structures is locally optimal over the construct. However, not all combinations can be induced by a coalition structure over the construct. Under a given coalition structure, for any pair of nodes $X$, let $c(X)$ be the logical value of whether the two nodes in $X$ do not lie in the same coalition. Then, there are locally optimal coalition structures which give every possible combination of logical values for $c(A)$, $c(B)$ and $c(C)$ apart from $c(A)$, $c(B)$ and $c(C)$ all being false.
\begin{figure}[ht]
\subfigure[Construct structure]{
\label{sfig:construct2cons}
\includegraphics[scale = 0.8]{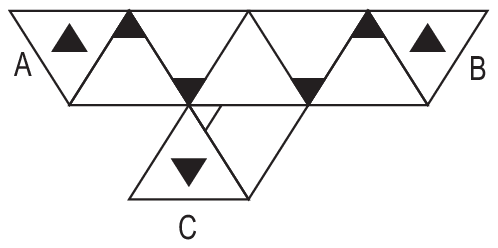}}
\subfigure[Partial locally optimal structure]{
\label{sfig:construct2Astate1}
\includegraphics[scale = 0.8]{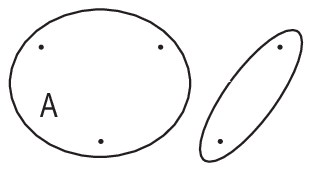}}
\subfigure[Partial locally optimal structure]{
\label{sfig:construct2Astate2}
\includegraphics[scale = 0.8]{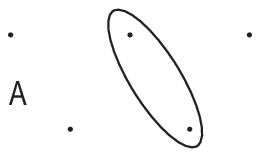}}
\subfigure[Partial locally optimal structure]{
\label{sfig:construct2Bstate1}
\includegraphics[scale = 0.8]{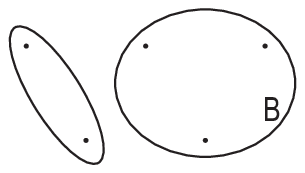}}
\subfigure[Partial locally optimal structure]{
\label{sfig:construct2Bstate2}
\includegraphics[scale = 0.8]{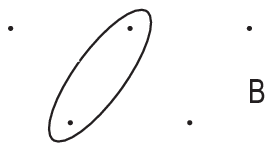}}
\subfigure[Partial locally optimal structure]{
\label{sfig:construct2Cstate1}
\includegraphics[scale = 0.8]{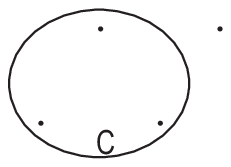}}
\subfigure[Partial locally optimal structure]{
\label{sfig:construct2Cstate2}
\includegraphics[scale = 0.8]{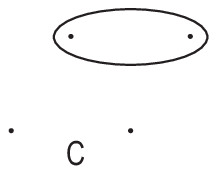}}
\subfigure[Partial locally optimal structure]{
\label{sfig:construct2Cstate3}
\includegraphics[scale = 0.8]{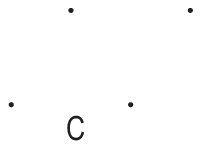}}
\caption{Edge sum construct.}\label{fig:construct2}
\end{figure}
We construct our edge sum problem to represent the 3-SAT problem as follows. We create a copy of the construct in Figure~\ref{fig:construct2} for each clause of the problem. The three pairs labelled $A, B, C$ are identified with the three literals in the corresponding clause. We identify a coalition structure over these constructs with a set of logical values for the literals in the clauses by saying that the literal associated with a pair of node is set as true if and only if those nodes are not in the same coalition. Using the component in Figure~\ref{fig:symbol5} we can connect the pairs of nodes that represent literals of the same variable or its negation to a series of copies of the constructs in Figures~\ref{fig:construct3} and~\ref{fig:construct4}. This allows us to we can ensure that any locally optimal coalition structure assigns consistent logical values to variables. To ensure that the resulting graph is planar, we can replace any pair of components from Figure~\ref{fig:symbol5} which cross over with two copies of the construct in Figure~\ref{fig:construct1}. A locally optimal coalition structure exists if and only if the original 3-SAT problem is satisfiable, and given any locally optimal coalition structure, we can identify a solution to the 3-SAT problem. Furthermore, if a locally optimal coalition structure exists, then a coalition structure is optimal if and only if it is locally optimal. The size of this graph is $O(m^2)$ and thus the proof is complete.
\begin{figure}[ht]
\includegraphics[scale = 0.5]{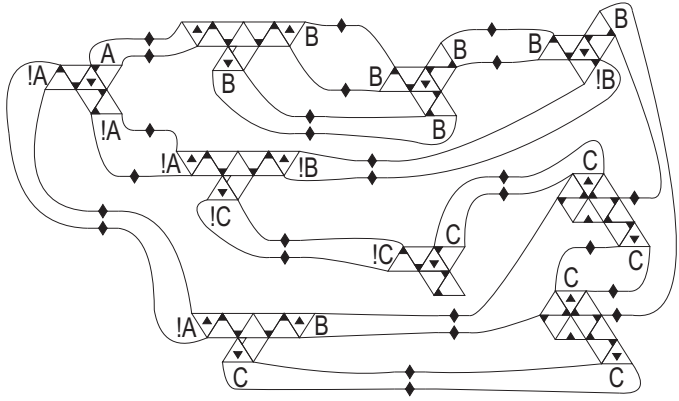}
\caption{Reduction of $(A \vee B \vee B) \wedge (!A \vee !B \vee !C) \wedge (!A \vee B \vee C)$.}\label{fig:examplereduction}
\end{figure}
An example of this reduction process is shown in Figure~\ref{fig:examplereduction}
for the 3-SAT problem
$(A \vee B \vee B) \wedge (!A \vee !B \vee !C) \wedge (!A \vee B \vee C)$.
\end{proof}


\section{Conclusions}\label{sec:conclusions}
\noindent
This paper initiates the study of coalition structure generation over graphs (GCSG) and provides the foundation for analysis of its computational complexity. Our results show that the problem can be solved in polynomial time for small minor free graphs, but is NP--complete for general, and even for planar graphs, with simple edge sum valuation functions. Future research on this topic will include the study of approximability of the GCSG problem for these and other interesting graph classes, and developing approximation schemes where applicable.

\bibliographystyle{plain}
\bibliography{cfrefs}

\begin{thebibliography}{10}

\bibitem{MINORSEPARATOR}
N.~Alon, P.~Seymour, and R.~Thomas.
\newblock A separator theorem for graphs with an excluded minor and its
  applications.
\newblock In {\em Proceedings of 22nd ACM Symposium on Theory of Computing},
  pages 293--299, 1990.

\bibitem{ConSan2004}
V.~Conitzer and T.~Sandholm.
\newblock Computing shapley values, manipulating value division schemes, and
  checking core membership in multi-issue domains.
\newblock In {\em Proceedings of the 19th Conference on Artificial
  Intelligence}, pages 219--225, 2004.

\bibitem{ConSan2006}
V.~Conitzer and T.~Sandholm.
\newblock Complexity of constructing solutions in the core based on synergies
  among coalitions.
\newblock {\em Artificial Intelligence}, 170(6):607--619, 2006.

\bibitem{TREE}
G.~Demange.
\newblock On group stability in heirarchies and networks.
\newblock {\em Journal of Political Economy}, 112(4):754--778, 2004.

\bibitem{Ieong2005}
S.~Ieong and Y.~Shoham.
\newblock Marginal contribution nets: A compact representation scheme for
  coalitional games.
\newblock In {\em Proceedings of the 6th ACM Conference on Electronic Commerce
  (ACM EC)}, pages 193--202, 2005.

\bibitem{PLANARGRAPHSEPARATOR}
R.~J. Lipton and R.~E. Tarjan.
\newblock A separator theorem for planar graphs.
\newblock {\em Journal of Applied Mathematics}, 36(2):177--189, 1979.

\bibitem{COMPACTREPRESENTATION}
N.~Ohta, V.~Conitzer, R.~Ichimura, Y.~Sakurai, A.~Iwasaki, and M.~Yokoo.
\newblock Coalition structure generation utilizing compact characteristic
  funciton representations.
\newblock In {\em Proceedings of the 15th International Joint Conference on
  Principles and Practice of Constraint Programming}, pages 623--638, 2009.

\bibitem{Rahwan2008}
T.~Rahwan and N.~R. Jennings.
\newblock An improved dynamic programming algorithm for coalition structure
  generation.
\newblock In {\em Proceedings of the 7th International Conference on Autonomous
  Agents and Multi-Agent Systems}, pages 1417--1420, 2008.

\bibitem{Rahwan2007}
T.~Rahwan, S.~D. Ramchurn, V.~D. Dang, A.~Giovannucci, and N.~R. Jennings.
\newblock Anytime optimal coalition structure generation.
\newblock In {\em Proceedings of the 22nd Conference on Artificial Intelligence
  (AAAI)}, pages 1184--1190, 2007.

\bibitem{Rothkopf1998}
M.~H. Rothkopf, A.~Peke\v{c}, and R.~M. Harstad.
\newblock Computationally manageable combinatorial auctions.
\newblock {\em Management Science}, 44(8):1131--1147, 1998.

\bibitem{CSGWORSTCASE}
T.~Sandholm, K.~Larson, M.~Andersson, O.~Shehory, and F.~Tohm\'e.
\newblock Coalition structure generation with worst case guarantees.
\newblock {\em Artificial Intelligence}, 111(1-2):209--238, 1999.

\bibitem{Kraus1998}
O.~Shehory and S.~Kraus.
\newblock Methods for task allocation via agent coalition formation.
\newblock {\em Artificial Intelligence}, 101(1-2):165--200, 1998.

\bibitem{CHOOSEBOUND}
P.~St\v{a}nic\v{a}.
\newblock Good lower and upper bounds on binomial coefficients.
\newblock {\em Journal of Inequalities in Pure and Applied Mathematics},
  2(3):Art. 30, 2001.

\bibitem{COMPLETEMINOREDGES}
A.~Thomason.
\newblock The extremal function for complete minors.
\newblock {\em Journal of Combinatorial Theory, Series B}, 81(2):318--338,
  2001.

\bibitem{COMPLETESETPARTITIONING}
D.~Y. Yeh.
\newblock A dynamic programming approach to the complete set partitioning
  problem.
\newblock {\em BIT Numerical Mathematics}, 26(4):467--474, 1986.

\end{thebibliography}















\end{document}